\documentclass[runningheads]{llncs}
\usepackage[T1]{fontenc}
\usepackage{graphicx}
\usepackage{hyperref}
\usepackage{color}

\usepackage{amssymb}
\usepackage{amsmath}
\usepackage{xcolor}
\usepackage{thmtools}
\usepackage{apptools}
\usepackage{microtype}
\usepackage{paralist}
\usepackage{xspace}
\usepackage{subcaption}

\usepackage{cleveref}

\crefname{figure}{Figure}{Figures}
\crefname{theorem}{Theorem}{Theorems}
\crefname{lemma}{Lemma}{Lemmas}
\crefname{corollary}{Corollary}{Corollaries}
\crefname{section}{Section}{Sections}
\crefname{appendix}{Appendix}{Appendices}
\crefname{claim}{Claim}{Claims}

\graphicspath{{./figures/}}

\usepackage{enumitem}

\usepackage{apptools}
\newcommand{\restateref}[1]{\IfAppendix{\hyperref[#1]{$\star$}}{\hyperref[#1*]{$\star$}}}

\NewCommandCopy{\llncsqed}{\qed}
\newif\ifqedused
\RenewDocumentCommand{\qed}{}{%
    \ifmmode
        \global\qedusedtrue\rlap{\quad\llncsqed}%
    \else
        \ifqedused\else\llncsqed\fi\global\qedusedfalse
    \fi
}
\AtEndEnvironment{proof}{\hfill \ensuremath{\qed}~~~~~~}

\newcommand{\PPG}{path pair graph\xspace}
\newcommand{\xperm}{\pi_x}
\newcommand{\yperm}{\pi_y}

\newcommand{\xpath}{P_x}
\newcommand{\ypath}{P_y}

\newcommand{\naesatShort}{\textsc{NAE3SAT}\xspace}

\begin{document}
\title{On Simultaneous Embeddings of Two Paths on the Grid}

\author{Stephen Kobourov\inst{1}\orcidID{0000-0002-0477-2724} \and
    William Lenhart\inst{2}\orcidID{0000-0002-8618-2444} \and
    Giuseppe Liotta\inst{3}\orcidID{0000-0002-2886-9694} \and
    Daniel Perz\inst{4}\orcidID{0000-0002-6557-2355} \and
    Pavel Valtr\inst{4}\orcidID{0000-0002-3102-4166} \and
    Johannes Zink\inst{1}\orcidID{0000-0002-7398-718X}}
\authorrunning{S.\ Kobourov, W.\ Lenhart, G.\ Liotta, D.\ Perz, P.\ Valtr, and J.\ Zink}

\institute{Technical University of Munich, Heilbronn, Germany \and
    Williams College, Williamstown, MA, USA \and
    University of Perugia, Perugia, Italy \and
    Charles University, Prague, Czech Republic}

\maketitle

\begin{abstract}
Simultaneous geometric embedding of graphs on an integer grid is a classical research topic in Graph Drawing.
When the graphs are two paths, one can additionally require that
one path is $x$-monotone and the other is $y$-monotone (the \emph{monotone} setting).
A simple and elegant algorithm is known to simultaneously embed two paths on an integer grid in the monotone setting.
The resulting embedding is, however, not optimal in terms of grid size or edge lengths.
In this paper, we study some quality measures of simultaneous embeddings of two paths on an integer grid, in both the monotone and in the unconstrained setting.
In the monotone setting, we present almost-linear time algorithms to minimize either the perimeter of the bounding box or the sum of the edge lengths in the $L_1$-norm.
Both results are based on transforming the problem into that of finding a minimum  vertex cover (weighted for the $L_1$-norm problem) on a related bipartite graph derived from the two paths.
As for the unconstrained setting, we show that both minimizing the length of the longest edge and minimizing the sum of the edge lengths of a simultaneous embedding of two paths on an integer grid is NP-hard.

\end{abstract}

\section{Introduction}

A simultaneous geometric embedding of a pair of planar graphs sharing the same vertex set consists of two planar straight-line drawings of the two graphs
over the same set of points for the vertices.
If an edge is shared by the two graphs, then it is represented by the same straight-line segment, which makes it easy to visually discover common patterns in the drawings. 

The notion of simultaneous geometric embedding was introduced by
Brass et al.~\cite{DBLP:journals/comgeo/BrassCDEEIKLM07}
twenty years ago.
Since then, it has received considerable attention in the graph drawing literature,
partly because of the elegant theoretical challenges it poses and
partly because it lies at the core of dynamic graph visualization.
The survey by Bl{\"{a}}sius et al.~\cite{DBLP:reference/crc/BlasiusKR13} provides
extensive discussions about simultaneous geometric embeddings and their variants;
in particular, Section 11.8 discusses the connection to dynamic graph visualization.

For an application on the theoretical side,
the simultaneous embedding of cycles can be used to show max-degree-4 graphs have geometric thickness two~\cite{duncan2004geometric}.
Geometric thickness is a problem related to simultaneous geometric embedding
where the task is to find a straight-line drawing of a graph
and an assignment of its edges to the minimum number of parts
such that each part is drawn crossing-free.
In a sense, simultaneous geometric embedding is a relaxed version
of geometric thickness where we are given the partition and only need to find
a suitable assignment of vertices to points in the plane.

Unsurprisingly,
only very restricted pairs of graphs admit a simultaneous geometric embedding. For example,
Angelini et al.~\cite{DBLP:journals/jgaa/AngeliniGKN12} show that there exist a path and a tree of depth four that do not have a simultaneous geometric embedding,
while Cabello et al.~\cite{DBLP:journals/jgaa/CabelloKLMSV11} prove that every tree and a matching admit a simultaneous geometric embedding, but there exist a planar graph and a matching that do not.
Estrella-Balderrama et al.~\cite{EstrellaBalderrama2008Simultaneous} show that deciding whether two planar graphs admit a simultaneous geometric embedding is NP-complete. 
Moreover, the problem for a linear number of planar graphs is hard in the existential theory of the reals~\cite{DBLP:journals/jgaa/CardinalK15,DBLP:journals/dcg/Kyncl11}.
This has later been strengthened
even to a constant number of 31 pairwise edge-disjoint
planar graphs~\cite{DBLP:journals/algorithmica/ForsterKMPTV26}
or 9360 paths~\cite{DBLP:journals/jgaa/Schaefer21a}.

\begin{figure}[t]
    \centering
    \begin{subfigure}{.48\linewidth}
        \centering
        \includegraphics[page=1]{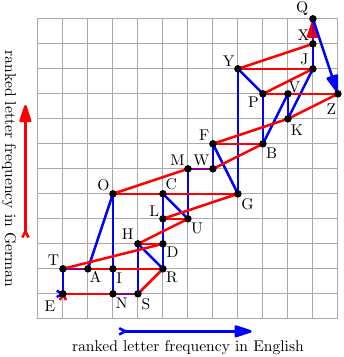}
        \caption{Ranking of the most frequent letters in English and German texts.}
        \label{fig:letters}
    \end{subfigure}
    \hfill
    \begin{subfigure}{.49\linewidth}
        \centering
        \includegraphics[page=2]{real-world-examples}
        \caption{25 most populous countries ranked by population and GDP per capita (2026).}
        \label{fig:countries}
    \end{subfigure}
    
    \caption{Real-world examples of two sequences (rankings) over the same data set.}
    \label{fig:real-world-example}
\end{figure}

On the other hand, Brass et al.~\cite{DBLP:journals/comgeo/BrassCDEEIKLM07} present an elegant linear-time algorithm to compute a geometric simultaneous embedding of two paths sharing the same $n$-vertex set:
Index both paths from 1 to~$n$ and
use these indices as $x$- and $y$-coordinates, respectively.
This yields a straight-line drawing on an $n\times n$ integer grid;
see \cref{fig:SPE-example} (left).
Note that such drawings possess the additional property
that the paths are monotone in $x$- and $y$-direction, respectively;
we call this the \emph{monotone} setting.

A natural motivation for this type of drawing
is to visualize two distinct sequences or rankings
over a common set of objects.
For an illustrative example, see \cref{fig:real-world-example}.
Roughly speaking,
if the drawing is close to
a diagonal from bottom left to top right,
then the two sequences are somewhat similar
(see \cref{fig:letters} vs.\ \cref{fig:countries}).
Also note that we do not require
strict monotonicity, i.e.,
vertices in the sequence drawn
in $x$-direction ($y$-direction) can
share the same $x$-coordinate ($y$-coordinate).

\paragraph{Our contribution.}
We revisit the problem of computing a simultaneous geometric embedding of two paths on an integer grid and consider the question of optimizing the drawing under various quality metrics.
We first consider the monotone setting (\cref{sec:monotone}), for which we present an almost-linear time algorithm to minimize the perimeter of the bounding box. 
Our approach is based on formulating an appropriate system of constraints that can be interpreted as a graph problem.
We then leverage the techniques used for minimizing the perimeter to develop an almost-linear time algorithm for computing  a monotone grid embedding that minimizes the sum of the edge lengths under the $L_1$-norm,
i.e., the sum of the $x$- and $y$-extents over all edges.

We also study the general version of the problem
where the paths are not required to be monotone (\cref{sec:nph}).
We show that the problem of finding a geometric simultaneous embedding of two paths on an integer grid is NP-hard when the objective is to minimize the maximum edge length, or to minimize the sum of edge lengths (both under the $L_1$-norm and the Euclidean norm).

For statements marked with ($\star$), we provide the full proof in the appendix.

\section{Monotone Setting}
\label{sec:monotone}

In this section we present a unified approach for efficiently solving the problems of minimizing either the perimeter or the total edge length in the $L_1$-norm for a simultaneous embedding of two directed paths~$\xpath$ and~$\ypath$ on an integer grid in the monotone setting.
We first introduce some notation.

The \emph{\PPG} of $\xpath$ and $\ypath$ is given by $G = G(\xpath,\ypath) =(V,E)$, where $V = \{v_1, \ldots, v_n \}$ is the common vertex set of $\xpath$ and $\ypath$, and $E$ is the union of the sets of directed edges of $\xpath$ and~$\ypath$.
Given two paths $\xpath$ and $\ypath$ whose \PPG is $G = (V, E)$,
we call a drawing of $G$ such that each $v \in V$ lies on an integer grid point, $\xpath$ is weakly $x$-monotone, and $\ypath$ is weakly $y$-monotone,
a {\em weakly monotone grid embedding of~$G$ (a WMGE of $G$)}.

Note that $\xpath$ and $\ypath$ induce permutations $\xperm$ and $\yperm$ on $V$ in which $\xperm(i)$ is the $i^\textrm{th}$ vertex of $\xpath$ and  $\yperm(i)$ is the $i^\textrm{th}$ vertex of $\ypath$.
Thus $E(G) = E = \{(v_{\xperm(i)},v_{\xperm(i+1)}) : 1 \leq i < n\}
\cup
\{(v_{\yperm(i)},v_{\yperm(i+1)}) : 1 \leq i < n\}$.
A WMGE of $G$ consists of assigning to each $v_k$ a distinct location $S_k = (x_k,y_k)$ on the integer grid such that $x_{\xperm(i)} \leq x_{\xperm(i+1)}$ and $y_{\yperm(i)} \leq y_{\yperm(i+1)}$ for all $1 \leq i < n$, and any two edges intersect in at most $1$ point (no ``segment overlap'').

We denote the predecessor and successor of $v$ on $\xpath$ by $p_x(v)$ and $s_x(v)$, respectively; $p_y(v)$ and $s_y(v)$ are defined analogously for path $\ypath$.
An internal vertex $v$ of $\xpath$ is a {\em switch vertex of $\xpath$} if either
\begin{inparaenum}[(i)]
    \item both $p_x(v)$ and $s_x(v)$ come before $v$ on $\ypath$ or
    \item both come after $v$ on $\ypath$.
\end{inparaenum}
The pair of edges $(p_x(v),v)$ and $(v, s_x(v))$ are called a {\em switch pair of $\xpath$}.
A switch vertex and switch pair of $\ypath$ are defined analogously.
A pair of vertices $\{u,v\}$ forms
a {\em shared edge}
if they are consecutive on both paths (although their relative order on each path may be different).
In a grid embedding, the $x$-extent ($y$-extent) of an edge is the absolute value of the difference in the $x$-coordinates ($y$-coordinates) of its endpoints.
Any WMGE of $G$ must satisfy:
\begin{enumerate}[label=(V\arabic*)]
    \item For each switch vertex $v$ of $\xpath$, the $x$-coordinates of $p_x(v)$ and $s_x(v)$ differ by at least $1$ (no vertical segment overlap),
    \item for each switch vertex $v$ of $\ypath$, the $y$-coordinates of $p_y(v)$ and $s_y(v)$ differ by at least $1$ (no horizontal segment overlap), and
    \item for any edge $\{v_i,v_j\}$ (for some $i, j \in \{1, \dots, n\}$), the sum of the $x$- and $y$-extents of the line segment $\overline{S_i S_j}$
    is at least $1$ (no two points in the same location).
\end{enumerate}

Note that in any WMGE of $G$, each embedded edge of $G$ will have non-negative integer $x$- and $y$-extents;
we denote the $x$- and $y$-extents of an embedded edge $e$ by $d_x(e)$ and $d_y(e)$, respectively.
The three properties above can be reformulated as the following $x$- and $y$-extent constraints on the embedded edges of $\xpath$ and~$\ypath$.
By \cref{lem:constraints-for-wmge},
we show that it even suffices to restrict the third condition to shared edges.

\begin{enumerate}[label=(E\arabic*)]
    \item\label{item:extent-constraints-1} For any switch vertex $v$ on $\xpath$, we have $d_x(p_x(v),v) + d_x(v,s_x(v)) \geq 1$,
    \item\label{item:extent-constraints-2} for any switch vertex $v$ on $\ypath$, we have $d_y(p_y(v),v) + d_y(v,s_y(v)) \geq 1$, and
    \item\label{item:extent-constraints-3} for any shared edge $\{u,v\}$, we have $d_x(u,v) + d_y(u,v) \geq 1$.
\end{enumerate}
Any two non-negative integer-valued functions satisfying constraints \ref{item:extent-constraints-1}--\!\ref{item:extent-constraints-3} yield a WMGE.
\cref{lem:constraints-for-wmge} shows a way how to construct such a WMGE:

\begin{restatable}[\restateref{lem:constraints-for-wmge}]{lemma}{constraintsForWMGE}
    \label{lem:constraints-for-wmge}
    Let $\{\tilde{d}_x(e) : e \in \xpath\}$ and $\{\tilde{d}_y(e) : e \in \ypath\}$ 
    be any non-negative integer-valued functions satisfying extent constraints \ref{item:extent-constraints-1}--\!\ref{item:extent-constraints-3}.
    Then $\tilde{d}_x()$ and $\tilde{d}_y()$
    determine a WMGE of~$G$ as follows:
    \begin{enumerate}
        \item Let $x(v_{\xperm(1)}) = 0$;
        \item let $x(v_{\xperm(i+1)}) = \tilde{d}_x(v_{\xperm(i)},v_{\xperm(i+1)}) + x(v_{\xperm(i)})$, \quad for each $i: 1 \leq i < n$;
        \item let $y(v_{\yperm(1)})=0$; \quad and
        \item let $y(v_{\yperm(i+1)}) = \tilde{d}_y(v_{\yperm(i)},v_{\yperm(i+1)}) + y(v_{\yperm(i)})$, \quad for each $i: 1 \leq i < n$.
    \end{enumerate}
\end{restatable}

Note that we do not need to specify, for non-shared edges, $y$-extents of edges of $\xpath$ or $x$-extents of edges of $\ypath$; these will be implicitly determined by the fact that for any edge $\{u,v\}$ on one of the paths, there is a unique subpath of the other path between $u$ and $v$.

\subsection{Minimizing the perimeter of the bounding box}
Consider a WMGE of the \PPG~$G$.
The perimeter of the bounding box is just twice the sum of the $x$-extents of the edges in $\xpath$ plus twice the $y$-extents of the edges in $\ypath$, and so does not depend on the assignment of an $x$-coordinate to $v_{\xperm(1)}$  and a $y$-coordinate to $v_{\yperm(1)}$.

Now note that in a minimum perimeter WMGE of $G$, each $d_x$ and $d_y$ value is either~$0$ or~$1$, since replacing an integer value greater than $1$ with $1$ does not violate any of the extent constraints.
Thus the problem of finding a minimum perimeter grid embedding of $G$ becomes the problem of finding $0 / 1$ values $\tilde{d}_x(e)$ for each edge $e$ of $\xpath$ and $\tilde{d}_y(e)$ for each edge $e$ of $\ypath$ that satisfy the three extent constraints and that minimize
$\sum_{e \in \xpath} \tilde{d}_x(e) +
\sum_{e \in \ypath} \tilde{d}_y(e)$.

\begin{figure}[t]
    \centering
    \includegraphics[page=10]{figures/SPE_example.pdf}
    \caption{A WMGE of a \PPG $G$ (left) and the constraint graph $C_G$ of~$G$ (right).
        The paths $\xpath$ and~$\ypath$ of $G$ and their vertices in $C_G$ are red and blue, respectively; shared edges are violet.}
    \label{fig:SPE-example}
\end{figure}

\begin{figure}[t]
    \centering
    \includegraphics[page=11]{figures/SPE_example.pdf}
    \caption{A WMGE of the \PPG $G$ (\cref{fig:SPE-example}) with minimum perimeter and the corresponding minimum vertex cover in $C_G$ (square boxes). Every vertex in the vertex cover corresponds to an edge whose endpoints differ by one in the $x$-direction or $y$-direction, respectively.}
    \label{fig:SPE-example_minimized}
\end{figure}

To compute a WMGE of a \PPG $G=(V,E)$, we construct a new graph $C_G = (V_C, E_C)$ from $G$, called the {\em constraint graph of $G$}, see \cref{fig:SPE-example} (right).
Let $X = \{ e : e \in \xpath\}$ and $Y = \{ e : e \in \ypath\}$.
The vertices and edges of~$C_G$ are defined as follows:
\begin{itemize}
    \item $V_C = V_X \cup V_Y$, where $V_X$ has a distinct vertex for each $e \in X$ and $V_Y$ has a distinct vertex for each $e \in Y$. Shared edges of $G$ therefore appear twice in~$V_C$.
    \item $E_C = E_X \cup E_Y \cup E_{XY}$ where
    \begin{itemize}
        \item $E_X$ is the set of all switch pairs for $\xpath$,
        \item $E_Y$ is the set of all switch pairs for $\ypath$, and
        \item $E_{XY}$ is a matching connecting each shared edge in $V_X$ to its counterpart in~$V_Y$.
    \end{itemize}
\end{itemize}

Note that the induced subgraphs of $C_G$ obtained by restricting the vertex set to $V_X$ and~$V_Y$, respectively, are each linear forests,
i.e., forests of paths.
It is easy to see that a WMGE of $G$ corresponds to a vertex  cover of $C_G$ since any satisfying assignment of $0 / 1$ values to $\tilde{d}_x()$ and $\tilde{d}_y()$ corresponds to an assignment of $0 / 1$ values to the vertices of $C_G$ such that every edge of $C_G$ has at least one vertex with value~$1$.
The converse also holds.
Further, the size of the cover gives (half) the perimeter of the embedding, so we have transformed the embedding problem into the problem of finding a minimum vertex cover of $C_G$, see \cref{fig:SPE-example_minimized}.

\begin{lemma}
    \label{lem:bipartite}
    For any \PPG $G$, $C_G$ is a bipartite graph.
\end{lemma}

\begin{proof}
    We show that any cycle in $C_G$ must have even length.
    Note that for any edge $e = \{u,v\}$ in $G$,  $u$ and $v$ have two induced orderings determined by $\xperm$ and~$\yperm$.
    We say $e$ has {\em positive alignment} if those orders are the same and {\em negative alignment} otherwise.
    Note that any two edges that form a switch pair for one of the two paths will have opposite alignments.
    Note also that this allows us to assign an alignment value to each vertex of $C_G$ and that the end vertices of an edge of $C_G$ will have equal alignments exactly when those vertices represent the same shared edge of $G$.
    
    Consider a cycle $C = v_1, \ldots, v_k$ in $C_G$. Each consecutive pair of vertices corresponds to either a switch pair or a shared edge of $G$.
    $C$ must cross between $V_X$ and $V_Y$ an even number of times, since $V_X$ and $V_Y$ induce linear forests in $C_G$.
    Thus there is an even number of edges of $C$ whose end vertices have the same alignment. Hence, $C$ is an even length cycle.
\end{proof}

Note that $C_G$ has size linear in the number~$n$ of vertices in~$G$
and can be constructed in linear time by traversing~$G$.
Since $C_G$ is bipartite, a minimum vertex cover of $C_G$ can be found by first determining a maximum matching,
which can be done in a bipartite graph straight-forwardly by means of a maximum flow.
Employing the recent maximum flow algorithm by van den Brand al.~\cite{DBLP:conf/focs/Brand0PKLGSS23},
we find a maximum matching in $O(|E_C|^{1 + o(1)}) \subseteq n^{1 + o(1)}$ time.
Given a maximum matching~$M$, we use the method of Kőnig~\cite{Konig1931} to traverse $C_G$ in linear time
building $M$-alternating paths in order to compute a minimum vertex cover~\cite[Theorem~5.3]{DBLP:books/others/BondyM76}.
We conclude with the following theorem.

\begin{theorem}
    There is an algorithm that,
    given a \PPG $G = G(\xpath,\ypath)$,
    computes a weakly monotone grid embedding of $G$ with minimum perimeter
    in $n^{1+o(1)}$ time where $n = |V(G)|$.
\end{theorem}

Note that, instead of using a very involved flow algorithm~\cite{DBLP:conf/focs/Brand0PKLGSS23},
we can easily find a maximum matching in
$O(|E_C| \sqrt{|V_C|}) \subseteq O(n^{3/2})$ time using the Hopcroft-Karp algorithm~\cite{Hopcroft1973}.

\subsection{Minimizing the sum of the $L_1$ edge lengths}\label{sse:L1-norm}

\begin{figure}[t]
    \centering
    \includegraphics[page=12]{figures/SPE_example.pdf}
    \caption{A WMGE of $G$ (\cref{fig:SPE-example}) minimizing the sum of the $L_1$ edge lengths and the corresponding minimum vertex cover in $C_G$ (square boxes).
        Each vertex in the vertex cover corresponds to an edge whose endpoints differ by one in the $x$-direction or $y$-direction, respectively.}
    \label{fig:SPE-example_L1}
\end{figure}

The approach used to minimize the perimeter of a WMGE can be adjusted to minimize the sum of the $L_1$ edge lengths.
The length of a segment in the~$L_1$-norm, also known as \emph{Manhattan norm}, is the sum of the $x$- and $y$-differences of its endpoints.
Let $\Gamma$ be a WMGE of a graph $G$ and denote the sum of the $L_1$ edge lengths of all edges in $\Gamma$ by $L_1(\Gamma)$.
Then,
\[ L_1(\Gamma) = \sum_{e \in \Gamma} (d_x(e) + d_y(e)). \]

As with perimeter, for a grid embedding that minimizes $L_1(\Gamma)$, it must be that for any edge $e$ of $\xpath$, $d_x(e)$ is either $0$ or $1$, as is $d_y(e)$ for any edge $e$ on~$\ypath$.
In order to take advantage of the approach of the previous section, we first rewrite the sum $L_1(\Gamma)$ so that it depends only on the $x$-extents of the edges of $\xpath$ and the $y$-extents of the edges of $\ypath$.

For any unshared edge $e$ of $\xpath$, let $\ypath(e)$ denote the unique subpath of $\ypath$ between the endpoints of $e$.
Note that then $d_y(e) = \sum_{e' \in \ypath(e)} d_y(e')$.
Similarly for any unshared edge~$e$ of $\ypath$, let $\xpath(e)$ denote the unique subpath of $\xpath$ between the endpoints of $e$, and so $d_x(e) = \sum_{e' \in \xpath(e)} d_x(e')$.
These sums can be substituted into $L_1(\Gamma)$ so that every term is now either an $x$-extent of an edge of $\xpath$ or a $y$-extent of an edge of $\ypath$.
For our purposes, we define $\xpath(e) = \ypath(e) = \emptyset$ if $e$ is a shared edge.

Besides the corresponding unique subpath of an edge~$e$ described by $\xpath(e)$ and~$\ypath(e)$,
we also want to refer, in the other direction,
to all edges for which an edge lies on such a corresponding unique subpath.
To this end, we introduce the sets~$C_e$.
For each edge $e \in \xpath$, let $C_e$ be the set of edges on $\ypath$ whose unique subpath of $\xpath$ contains~$e$, that is,
$C_e = \{e' \in \ypath \colon e \in \xpath(e') \}$.
Now define the {\em weight} of $e$ by $c_e = |C_e| + 1$. This counts how many times the term $d_x(e)$ appears in $L_1(\Gamma)$.
Symmetrically, for each edge $e \in \ypath$, let $C_e$ be the set of edges on $\xpath$ whose unique subpath of $\ypath$ contains~$e$, and
let $c_e = |C_e| + 1$; this counts how many times the term $d_y(e)$ appears in $L_1(\Gamma)$.
Thus,
\[ L_1(\Gamma) = \sum_{e \in \xpath} c_e d_x(e) + \sum_{e \in \ypath} c_e d_y(e). \]
Note that the weights $c_e$ are determined by the graph $G$, not the particular embedding.
We show in the appendix that we can find all weights in linear time.

\begin{restatable}[\restateref{lem:weight-linear-time}]{lemma}{weightLinearTime}
    \label{lem:weight-linear-time}
    The set of weights $\{c_e : e \in E(G)\}$ can be determined in linear time.
\end{restatable}

So we are seeking $0/1$ values $\tilde{d}_x(e) : e \in \xpath$ and $\tilde{d}_y(e) : e  \in \ypath$ that satisfy the three edge extent constraints~\!\ref{item:extent-constraints-1}--\!\ref{item:extent-constraints-3} while minimizing $L_1(\Gamma)$.
However, this is exactly the problem of finding a minimum weight vertex cover in the constraint graph $C_G = (V_C, E_C)$ of $G$, where each vertex $e$ in $V_C$ has non-negative integer weight $c_e$.

Since $C_G$ is bipartite (see \cref{lem:bipartite}) and has $O(n)$ edges, a minimum weight vertex cover can be computed in polynomial time via network flow techniques\footnote{%
    Computing a minimum weight vertex cover via minimum cuts, resp., maximum flows is a folklore result,
    which appears in many exercises of algorithms' courses.
    For an example, see Section~4 in \url{https://www.cs.dartmouth.edu/~deepc/LecNotes/cs31/lec24+25.pdf}
    by Deeparnab Chakrabarty.
}
in $n^{1 + o(1)}$ time~\cite{DBLP:conf/focs/Brand0PKLGSS23}.
In \cref{fig:SPE-example_L1} the sum of the $L_1$ edge lengths of the graph from \cref{fig:SPE-example} are minimized.

We conclude with the following theorem.

\begin{theorem}
    There is an algorithm that,
    given a \PPG $G = G(\xpath,\ypath)$, computes a weakly monotone grid embedding of $G$
    with minimum sum of edge length in the $L_1$-norm in $n^{1+o(1)}$ time where $n = |V(G)|$.
\end{theorem}

\begin{figure}[t]
    \centering
    \includegraphics[page=13]{figures/SPE_example.pdf}
    
    \caption{Path pair graph where the (unique) drawing
        minimizing the edge lengths under the $L_1$-norm (left) is different from
        the (unique) drawing minimizing perimeter (right).
        In the constraint graph,
        the minimum weighted and unweighted
        vertex cover (square boxes) differ.
        Left: perimeter 16, $\sum L_1$ lengths 28;
        right: perimeter 12, $\sum L_1$ lengths 30.}
    \label{fig:perimeter-L1-different}
\end{figure}

\begin{remark}
    \cref{fig:SPE-example_L1} also minimizes the perimeter.
    In general, however, both objectives can lead to distinct WMGEs.
    For an example, see \cref{fig:perimeter-L1-different}.
\end{remark}

\section{Unconstrained Setting}
\label{sec:nph}

If we drop the requirement that each path must be monotone,
the problem of computing their simultaneous embedding changes dramatically.
In fact,  deciding whether it is possible to simultaneously embed two given paths on an integer grid with unit edge lengths is NP-hard.
This implies that minimizing the maximum edge length and minimizing the sum of the edge length is also NP-hard in the unconstrained setting
for the $L_1$-norm and the Euclidean norm (\cref{cor:nph-min-total-e-length-min-max-e-length}).

\begin{figure}[p]
    \centering
    \includegraphics[page=6,width=\linewidth]{figures/hardness-1}
    \caption{Simultaneous embedding of two paths (red and blue; shared edges violet) with unit edges lengths
        that we obtain in our NP-hardness reduction from the \textsc{NotAllEqual3SAT} instance above.
        The terminal vertices of the two paths are drawn as square boxes.}
    \label{fig:nph}
\end{figure}

\begin{figure}[t]
    \centering
    \includegraphics[page=5,width=\linewidth]{figures/hardness-1}
    \caption{The two paths of \cref{fig:nph} shown individually.
    The terminal vertices of the two paths are drawn as square boxes.}
    \label{fig:nph1}
\end{figure}

\begin{restatable}[\restateref{thm:nph}]{theorem}{nph}
    \label{thm:nph}
    Given two paths on the same vertex set,
    it is NP-complete to decide whether there exists
    a simultaneous embedding on an integer grid that uses only unit edge lengths.
\end{restatable}

\begin{proof}[Sketch]
    Containment in NP is clear.
    Coordinates for all vertices serve as
    a polynomial-size certificate.
    We can check in polynomial time
    that every edge has length~1
    and that there are no overlaps of vertices. 
    
    We show NP-hardness by reduction from \textsc{NotAllEqual3SAT} (or short \naesatShort).
    In \naesatShort,
    we are given a Boolean formula
    in conjunctive normal form
    with at most three literals per clause
    and the task is to decide whether
    there is an assignment of Boolean values
    to the variables such that every clause
    contains at least one true and at least one false literal.
    It is well known that \naesatShort
    is NP-complete~\cite{s-csp-STOC78}.
    Our reduction resembles the \emph{logic engine}~\cite{DBLP:journals/ipl/BhattC87,DBLP:journals/tcs/EadesW96}.
    For an illustration of our construction see \cref{fig:nph,fig:nph1}.
    
    \begin{figure}[t]
        \centering
        \begin{subfigure}{\linewidth}
            \centering
            \includegraphics[page=7]{figures/hardness-1}
            \caption{two-short flag}
            \label{fig:nph:two-short-flag}
        \end{subfigure}
        
        \smallskip
        
        \begin{subfigure}{\linewidth}
            \centering
            \includegraphics[page=8]{figures/hardness-1}
            \caption{long flag}
            \label{fig:nph:long-flag}
        \end{subfigure}
        
        \begin{subfigure}{\linewidth}
            \centering
            \includegraphics[page=9]{figures/hardness-1}
            \caption{shaft part}
            \label{fig:nph:shaft-part}
        \end{subfigure}
        \caption{Building blocks in our NP-hardness construction.
        Green background indicates 4-cycles that can only be realized as unit squares as depicted.}
        \label{fig:nph:blocks}
    \end{figure}
    
    Let the input \naesatShort instance have $n$ variables and $m$ clauses.
    On a high level, we have a rigid frame on the left and on the right.
    In the middle there is a \emph{shaft} connecting the two parts of the frame.
    The shaft consists of $n$ rigid \emph{shaft parts} connected by single edges;
    see \cref{fig:nph:shaft-part}.
    Each such shaft part belongs to a \emph{striker};
    in \cref{fig:nph}, there are four vertically extending strikers.
    Since each striker is attached to the rest by two single edges,
    there is one degree of freedom to flip the drawing of each striker individually
    (in \cref{fig:nph}, each striker might be mirrored
    along a horizontal axis through the shaft).
    Each striker corresponds to a variable
    and its one drawing represents the assignment true,
    while its flipped drawing (depending on whether its shaft part is in the state up or down) represents the assignment false.
    For every clause, there are two horizontal strips (called \emph{clause strips}; colored in \cref{fig:nph}),
    one above and one below the shaft.
    Each striker has $4m$ \emph{flags}~-- one per strip and a ``dummy'' flag between two clause strips.
    Each flag is (in its striker) vertically connected to the neighboring flag or the shaft by a single edge
    and, thus, can be flipped to the left or right individually.
    A flag is internally rigid because the union of its two paths form mainly
    4-cycles which can only be realized as unit squares.
    A flag extends to the left and right, and
    it can have two short sides (\emph{two-short flag}, \cref{fig:nph:two-short-flag})
    or a short and a long side (\emph{long flag}, \cref{fig:nph:long-flag}).
    All two-short flags are isomorphic
    and all long flags are isomorphic.
    All flags in the clause strips are long flags unless
    the variable appears in the corresponding clause.
    
    If the variable appears normally (negated) in the clause,
    then a two-short flag is in the clause strip above the shaft
    when the orientation of the striker corresponds to true (false).
    All ``dummy'' flags are two-short flags
    and their orientation is not further relevant;
    they are just there for keeping space.
    For a drawing without crossings,
    the long side of a striker must be drawn in one of $n-1$ gaps between two strikers
    (next to the frame, there is no gap that can fit the long side of a flag).
    Two long sides of flags of neighboring striker cannot be drawn in the same gap.
    This means that in both clause strips of each clause,
    there is at least one two-short flag.
    This corresponds to a truth assignment where at least one literal per clause
    is false and at least one literal per clause is true.
    
    We can show that there is a crossing-free drawing with unit lengths
    if and only if the input \naesatShort instance is a yes-instance.
\end{proof}

\begin{corollary}
    \label{cor:nph-min-total-e-length-min-max-e-length}
    Simultaneous embedding of two paths on an integer grid
    is NP-hard for the objectives of
    minimizing the maximum edge length and minimizing the sum of edge lengths
    (both under the $L_1$-norm and the Euclidean norm).
\end{corollary}

\section{Future Work}
Can our approach in the monotone setting be adjusted to minimize the length of the longest edge and the sum of edge lengths in the Euclidean norm?
Can our NP-hardness proof in the unconstrained setting be extended to hold also for other criteria like the perimeter?
Also studying the complexity of minimizing the area in both the monotone and in the unconstrained setting would be interesting.
Moreover, one can consider more complex graphs than paths or more than two paths
when optimizing various objectives of simultaneous geometric embeddings on the integer grid.

\begin{credits}
\subsubsection{\ackname} This research was initiated at the Eighteenth Bertinoro Workshop on Graph Drawing 2025.
\end{credits}

\bibliographystyle{splncs04}
\bibliography{references}

\clearpage
\appendix

\section*{Appendix}
Here, we present full proofs omitted in the main part.

\section{Omitted Proofs of \cref{sec:monotone}}

\constraintsForWMGE*
\label{lem:constraints-for-wmge*}

\begin{proof}
    Since the first vertices of $\xpath$ and $\ypath$ have an integer coordinate (namely~0) in $x$- and $y$-dimension, respectively,
    and after that we add only integer values from $\tilde{d}_x$ and $\tilde{d}_y$, respectively,
    all coordinates are integers.
    Because both functions are non-negative, $\xpath$ is weakly $x$-monotone and $\ypath$ is weakly $y$-monotone.
    
    Next we show that any two edges~$e_1$ and~$e_2$ intersect in at most one point, i.e., there is no overlapping.
    If $e_1$ and $e_2$ belong to the same path, then they can only overlap
    if they are both vertical and share the same $x$-coordinate (on $\xpath$)
    or they are both horizontal and share the same $y$-coordinate (on $\ypath$).
    Also when going from $e_1$ to~$e_2$ along $\xpath$ (or $\ypath$),
    some edges are oriented upwards while others must be oriented downwards
    (or leftwards vs.\ rightwards along~$\ypath$).
    This means, there is a switch vertex between $e_1$ and~$e_2$,
    which means that $e_1$ and~$e_2$ get distinct $x$-coordinates (or $y$-coordinates)
    due to \ref{item:extent-constraints-1} (or \ref{item:extent-constraints-2}).
    If $e_1$ belongs to~$\xpath$ and~$e_2$ belongs to~$\ypath$
    and there is an overlap,
    then an endpoint~$v$ of~$e_2$ lies in the interior of~$e_1$
    (if an endpoint of~$e_1$ lies in the interior of~$e_2$,
    a symmetric argument can be applied).
    Then, however, $v$ has an $x$-coordinate between
    the $x$-coordinates of the endpoints of~$e_1$ (which are consecutive along~$\xpath$),
    which violates the non-negativity of~$\tilde{d}_x$ or~\ref{item:extent-constraints-1}.
    More generally, the latter argument also implies that there is no vertex in the interior of an edge.
    
    Finally, we show that no two distinct vertices coincide on a grid point.
    Suppose for a contradiction that two vertices $u$ and $v$ are assigned to the same locations.
    Note that $\{u,v\}$ cannot form a shared edge, since  $\tilde{d}_x(u,v) + \tilde{d}_y(u,v) \geq 1$ holds for shared edges
    due to~\ref{item:extent-constraints-3}.
    Let $x_0$ be the common $x$-coordinate of $u$ and $v$; there must be two consecutive vertices $w$ and $w'$ of $\xpath$ that have $x$-coordinate~$x_0$.
    Then $\{w,w'\}$ is not an edge of $\ypath$, so there is a unique subpath~$\Pi$ of $\ypath$ between $w$ and $w'$ with at least one inner vertex.
    On $\xpath$, there must be an inner vertex~$x$ of $\Pi$
    that is first or last among all vertices of the $\Pi$.
    But then $x$ is a switch vertex of $\ypath$ and its predecessor and successor on $\ypath$ must have vertical distance at least~$1$.
    Therefore $u$ and $v$ also have vertical distance at least one, and are not in the same location.
\end{proof}

\weightLinearTime*
\label{lem:weight-linear-time*}

\begin{proof}
    We show how to determine the weight for each edge along the path~$\xpath$ in linear time
    given a graph representation of the \PPG $G(\xpath, \ypath)$.
    Computing the weights of the edges along~$\ypath$ can be done symmetrically.
    
    For any unshared edge~$e$ of $\ypath$, 
    the subpath $\xpath(e)$ contains a substring of the edges on~$\xpath$;
    see \cref{fig:weight-linear-time-a} for an illustration.
    When we traverse~$\xpath$, we reach at some point the first edge~$e'$ of~$\xpath(e)$.
    Now we know that $e$ contributes 1 to $c_{\hat{e}}$ for each $\hat{e} \in \xpath(e)$,
    that is, for~$e'$ and each following edge along~$\xpath$
    until we have reached the last edge~$e''$ of~$\xpath(e)$.
    Hence, it suffices to increase a counter by~1 at the left endpoint of~$e'$
    and to decrease that counter by~1 at the right endpoint of~$e''$
    when traversing~$\xpath$ from left to right.
    
    \begin{figure}[t]
        \centering
        \begin{subfigure}{\linewidth}
            \centering
            \includegraphics[page=1]{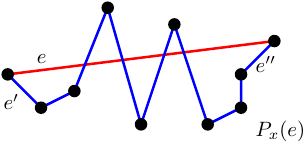}
            \subcaption{The edge $e \in \ypath$ (red) lies in $C_{\hat{e}}$
                for every edge $\hat{e} \in \xpath(e)$ (blue).
                Thus, $e$ contributes~1 to each $c_{\hat{e}}$.
                Since $\xpath(e)$ is a substring of~$\xpath$,
                we traverse~$\xpath$ once and
                we account for~$e$ by increasing
                a counter when reaching~$e'$
                and by decreasing that counter when leaving~$e''$,
                in order to compute the weight~$c_{\hat{e}}$
                of each edge $\hat{e} \in \xpath$.}
            \label{fig:weight-linear-time-a}
        \end{subfigure}
        
        \bigskip
        
        \begin{subfigure}{.23\linewidth}
            \centering
            \includegraphics[page=2]{lemma-compute-ce}
            \subcaption{counter: $+2$}
            \label{fig:weight-linear-time-b}
        \end{subfigure}
        \hfill
        \begin{subfigure}{.23\linewidth}
            \centering
            \includegraphics[page=3]{lemma-compute-ce}
            \subcaption{counter: $\pm0$}
            \label{fig:weight-linear-time-c}
        \end{subfigure}
        \hfill
        \begin{subfigure}{.23\linewidth}
            \centering
            \includegraphics[page=4]{lemma-compute-ce}
            \subcaption{counter: $-2$}
            \label{fig:weight-linear-time-d}
        \end{subfigure}
        \hfill
        \begin{subfigure}{.23\linewidth}
            \centering
            \includegraphics[page=5]{lemma-compute-ce}
            \subcaption{counter: $-1$}
            \label{fig:weight-linear-time-e}
        \end{subfigure}
        \caption{Illustration of the main ideas in \cref{lem:weight-linear-time}.}
        \label{fig:weight-linear-time}
    \end{figure}
    
    More precisely, we traverse $\xpath$ once and we maintain a counter.
    Initially, the counter is at 1.\footnote{%
    We initialize the counter with 1 instead of 0 because
    in the formula of the weight $c_e$ we always add $+ 1$.}
    Then, for $i \in \{1, \dots, n-1\}$, we check at~$\pi_x(i)$ (that is, the source point of the $i$-th edge of~$\xpath$)
    how many edges of~$\ypath$ extend to the left and how many extend to the right (there are at most two).
    Shared edges are not counted.
    For each edge~$e$ extending to the left, we decrease the counter by~1
    because~$e$ has $(\pi_x(i-1), \pi_x(i))$ as the last entry of~$\xpath(e)$;
    the edge $(\pi_x(i), \pi_x(i+1))$ is not included in~$\xpath(e)$.
    For each edge~$e$ extending to the right, we increase the counter by~1
    because~$e$ has $(\pi_x(i), \pi_x(i+1))$ as the first entry of~$\xpath(e)$.
    See \cref{fig:weight-linear-time-b,fig:weight-linear-time-c,fig:weight-linear-time-d,fig:weight-linear-time-e}
    for a examples of the net change of the counter at $\pi_x(i)$.
    
    After adjusting the counter, the current value of the counter is the weight of the edge $(\pi_x(i), \pi_x(i+1))$ and we save it.
    Overall, this is a linear number of constant-time operations, which shows the lemma.
\end{proof}

\section{Omitted Proofs of \cref{sec:nph}}

\nph*
\label{thm:nph*}

\begin{proof}
    Containment in NP is clear.
    Coordinates for all vertices serve as
    a polynomial-size certificate.
    We can check in polynomial time
    that every edge has length~1
    and that there are no overlaps of vertices. 
    
    \paragraph{Overview.}
    We show NP-hardness by reduction from \textsc{NotAllEqual3SAT}.
    In \textsc{NotAllEqual3SAT},
    we are given a Boolean formula
    in conjunctive normal form
    with at most three literals per clause
    and the task is to decide whether
    there is an assignment of Boolean values
    to the variables such that every clause
    contains at least one true and at least one false literal.
    It is well known that \textsc{NotAllEqual3SAT}
    is NP-complete~\cite{s-csp-STOC78}.
    Given a \textsc{NotAllEqual3SAT} instance~$F$,
    we define a pair of paths $(P_1, P_2)$ on the same vertex set~$V$,
    such that $F$ is a yes instance if and only if
    there is a simultaneous embedding of $(P_1, P_2)$
    on an integer grid with all edges having length~1
    (i.e., being horizontal or vertical integer grid segments).
    We refer to the graph with vertex set $V$
    and edge set being the union of the edge sets of $P_1$ and $P_2$,
    as the \PPG $G$.
    Our reduction resembles the \emph{logic engine}~\cite{DBLP:journals/ipl/BhattC87,DBLP:journals/tcs/EadesW96}.
    For an illustration of our construction see \cref{fig:nph}.
    Note that \cref{fig:nph1} shows the two individual paths.
    
    Let $F$ have $n$ variables and $m$ clauses.
    On a high level, a unit-edge-length drawing of $G$ has
    a rigid frame on the left and on the right.
    In the middle there is a \emph{shaft} connecting the two parts of the frame.
    The shaft consists of $n$ rigid \emph{shaft parts} connected by single edges;
    see \cref{fig:nph:shaft-part}.
    Each such shaft part belongs to a \emph{striker};
    in \cref{fig:nph}, there are 4 vertically extending strikers.
    Each striker is attached to the rest by two single edges
    at its shaft part~-- these two edges are $e_\mathrm{L}$ and~$e_\mathrm{R}$.
    There is one degree of freedom to flip the drawing of each striker individually up or down
    (in \cref{fig:nph}, each striker might be mirrored
    along a horizontal axis through the shaft).
    Each striker corresponds to a variable
    and its one drawing represents the assignment true,
    while its flipped drawing (depending whether its shaft part is in the state up or down;
    see \cref{fig:nph:shaft-part}) represents the assignment false.
    For every clause, there are two horizontal strips (called \emph{clause strips}; colored in \cref{fig:nph},
    one above and one below the shaft).
    Each striker has $4m$ \emph{flags} (see \cref{fig:nph:two-short-flag,fig:nph:long-flag})~--
    one per strip and a ``dummy'' flag between two clause strips.
    Each flag is (in its striker) vertically connected to the neighboring flag or the shaft by a single edge
    (called $e_\mathrm{top}$ and $e_\mathrm{bot}$ in \cref{fig:nph:two-short-flag,fig:nph:long-flag})
    and, thus, can be flipped to the left or right individually.
    A flag is internally rigid and extends to the left and right.
    It can have two short sides (\emph{two-short flag}, \cref{fig:nph:two-short-flag})
    or a short and a long side (\emph{long flag}, \cref{fig:nph:long-flag}).
    All two-short flags are isomorphic
    and all long flags are isomorphic.
    All flags in the clause strips are long flags unless
    the variable appears in the corresponding clause.
    
    If the variable appears normally in the clause,
    then a two-short flag is in the clause strip above the shaft
    when the orientation of the striker corresponds to true.
    If the variable appears negated in the clause,
    then a two-short flag is in the clause strip above the shaft
    when the orientation of the striker corresponds to false.
    All ``dummy'' flags are two-short flags
    and their orientation is not further relevant;
    they are just there for keeping space.
    For a drawing without crossings,
    the long side of a striker must be drawn in one of $n-1$ gaps between two strikers
    (next to the frame, there is no gap that can fit the long side of a flag).
    Two long sides of flags of neighboring striker cannot be drawn in the same gap.
    This means, in both clause strips of each clause,
    there is at least one two-short flag.
    This corresponds to a truth assignment where at least one literal per clause
    is false and at least one literal per clause is true.
    
    \paragraph{Basic building blocks.}
    Now we describe in more detail the basic building blocks of the \PPG $G$.
    We outline the important properties of these subgraphs of~$G$ in words, but,
    since the detailed connections are rather involved, we give the exact definitions by picture.
    A \emph{shaft part} is the subgraph defined in \cref{fig:nph:shaft-part}.
    It has 61 vertices and 6 shared edges connecting it to the rest of the graph:
    \begin{itemize}
        \item $e_\mathrm{L}$ and $e_\mathrm{R}$ are the shared edges connecting
        a shaft part to the left and right neighboring shaft part
        such that $e_\mathrm{L}$ and $e_\mathrm{R}$ are identical to
        the $e_\mathrm{R}$ and $e_\mathrm{L}$ of the left neighboring
        and the right neighboring shaft part, respectively.
        For the leftmost and rightmost shaft part, they connect to the frame instead.
        \item $e_1$ and $e_2$ connect to the sequence of flags of the striker on the one side of the shaft part.
        \item $e_3$ and $e_4$ connect to the sequence of flags of the striker on the other side of the shaft part.
    \end{itemize}
    We next argue about the rigidity of a shaft part.
    Note that a 4-cycle in~$G$ can be drawn with unit edge length only
    as a unit square with horizontal and vertical segments.
    Moreover, if two 4-cycles share an edge, they
    can only be drawn as adjacent unit squares with the same orientation.
    Hence, any set of connected 4-cycles can only be realized in one way
    (up to rotating by 90 degrees and mirroring, which we ignore for a moment).
    In \cref{fig:nph:shaft-part}, all connected 4-cycles have green background.
    If there exists a simultaneous embedding with unit edge lengths,
    they are drawn as depicted there.
    Now note that there is an additional path between $v_1$ and~$v_2$.
    Since it has length~5 and does not overlap any edge
    of a square with green background, it must be drawn as in \cref{fig:nph:shaft-part}.
    Then, there is an additional path between $v_3$ and~$v_4$ of length~3,
    which must be drawn as depicted.
    A similar argument can be applied to the remainder of the shaft part.
    Therefore (ignoring $e_\mathrm{L}$, $e_\mathrm{R}$, $e_1$, $e_2$, $e_3$, and $e_4$),
    a shaft part can only be drawn as in \cref{fig:nph:shaft-part}.
    Hence, it has width~14 and height~4.
    
    Two neighboring shaft parts must be aligned in the sense that they all have width~14 and height~4.
    They cannot flip rotation by 90 degrees because all $e_\mathrm{L}$, $e_\mathrm{R}$
    extend only to one direction~-- either horizontally (as in \cref{fig:nph:shaft-part}) or vertically.
    Suppose from now on they are aligned horizontally.
    Then, the only level of freedom a shaft part still has is mirroring
    along a horizontal axis.
    Note that this axis goes through all of $e_\mathrm{L}$ and $e_\mathrm{R}$.
    Thus, each shaft part can be drawn in one of two states called \emph{down} (\cref{fig:nph:shaft-part}, bottom left)
    or \emph{up} (\cref{fig:nph:shaft-part}, top right).
    Observe that the state of one shaft part does not depend on the states of the neighboring shaft parts.
    
    We have $n$ aligned shaft parts that together are called the \emph{shaft}.
    The rightmost shaft part is slightly different in the sense that it has only width~12 (see \cref{fig:nph}, right).
    The shaft is connected by an edge on the left and on the right to the \emph{frame} (see \cref{fig:nph}, left and right).
    The frame is a long collection of adjacent 4-cycles and hence rigid.
    As it connects by two single edges to the shaft, it can only be oriented as in \cref{fig:nph} with respect to the shaft.
    It extends far to the top and to the bottom such that
    it leaves a rectangular area of width $15n - 1$ between
    the left part of the frame and the right part of the frame.
    
    Next we define the \emph{flags}.
    A \emph{two-short flag} is the subgraph defined in \cref{fig:nph:two-short-flag}.
    It has 41 vertices and 2 shared edges ($e_\mathrm{top}$ and $e_\mathrm{bot}$)
    connecting it to the rest of the graph.
    The 4-cycles with green background can only be realized
    as rigid unit squares as depicted in \cref{fig:nph:two-short-flag}
    (up to rotation by 90 degrees and mirroring).
    By a similar argument as before,
    also the paths of length~3 and length~7
    consisting of (violet) shared edges within the two-short flag
    can only be realized as depicted.
    This establishes the rigidity of the two-short flag,
    which has width~8 and height~4
    (ignoring $e_\mathrm{top}$ and $e_\mathrm{bot}$).
    
    A \emph{long flag} is the subgraph defined in \cref{fig:nph:long-flag}.
    It has 57 vertices and 2 shared edges ($e_\mathrm{top}$ and $e_\mathrm{bot}$)
    connecting it to the rest of the graph.
    By a similar argument as for the two-short flag,
    a long flag is rigid and can only be realized as depicted.
    A long flag has width~12 and height~4
    (ignoring $e_\mathrm{top}$ and $e_\mathrm{bot}$).
    
    \paragraph{Composition.}
    Now we define the \emph{strikers}.
    Each striker represents a variable.
    Consider the striker representing the $i^\mathrm{th}$ variable (call it $x_i$) for some $i \in \{1, \dots, n\}$.
    The middle part of this striker is the $i^\mathrm{th}$ shaft part of the shaft.
    At $e_1$ of that shaft part, we add $2m$ flags in a sequence
    such that $e_1$ is identical to~$e_\mathrm{bot}$ of the first flag,
    whose $e_\mathrm{top}$ is identical to $e_\mathrm{bot}$ of the second flag,
    whose $e_\mathrm{top}$ is identical to $e_\mathrm{bot}$ of the third flag, etc.
    For each $j \in \{1, \dots, m\}$,
    the $(2j - 1)^\mathrm{th}$ flag is a two-short flag (previously called ``dummy'' flag)
    and the $(2j)^\mathrm{th}$ flag is a long flag
    unless $x_i$ appears unnegated in the $j^\mathrm{th}$ clause of~$F$,
    in which case it is also a two-short flag.
    
    Symmetrically, we add $2m$ flags in a sequence at $e_3$.
    Similarly, for each $j \in \{1, \dots, m\}$,
    the $(2j - 1)^\mathrm{th}$ flag is a two-short flag
    and the $(2j)^\mathrm{th}$ flag is a long flag
    unless $x_i$ appears, now, negated in the $j^\mathrm{th}$ clause of~$F$,
    in which case it is also a two-short flag.
    
    Consider the two outermost flags of the striker.
    We need to extend their $e_\mathrm{top}$ and $e_\mathrm{bot}$
    such that the whole instance has only two paths in total.
    To this end, we add a path of sufficiently many shared edges
    that terminate in $e_2$ and $e_4$ of the shaft part, respectively.
    Sufficient means here that we can reach $e_2$ or $e_4$ and we have
    some flexibility to wiggle around long and short flags.
    It is not a problem if this subpath is longer than needed
    because additional lengths can be placed
    on top/bottom of the drawing (see \cref{fig:nph}).
    We only need to make sure that the parity of the total length is odd.
    
    It remains to argue about the overall shape of a striker.
    We have already argued that the shaft part and the flags are (internally) rigid.
    Let the shaft part be drawn horizontally as in \cref{fig:nph,fig:nph:shaft-part}.
    The edge $e_1$ coincides with $e_\mathrm{bot}$ of a flag
    and it must be drawn vertically.
    If it was drawn horizontally, the flag and the shaft part would overlap.
    Thus, the first flag is drawn horizontally as in \cref{fig:nph,fig:nph:two-short-flag}.
    As its $e_\mathrm{top}$ coincides with the $e_\mathrm{bot}$ of the next flag,
    also that flag and all following flags extend horizontally;
    otherwise two neighboring flags would overlap.
    The symmetric argument applies to the flag attached at~$e_3$ and all following flags.
    Thus, there is a vertical axis containing $e_\mathrm{top}$ and $e_\mathrm{bot}$
    of all flags of a striker; we call it \emph{striker axis}.
    The only flexibility that remains within a flag is to mirror
    the drawing of the flag along the striker axis.
    This results in two states for the flags
    called \emph{left} and \emph{right} as shown in \cref{fig:nph:two-short-flag,fig:nph:long-flag}.
    As for the shaft parts, observe that the state of one flag
    is clearly independent of the states of the neighboring flags
    because each pair of neighboring flags is connected by a single shared edge.
    
    Crucially, a long flag extends 8 to the left and 4 to the right if it is in the state left,
    while a long flag extends 4 to the left and 8 to the right if it is in the state right.
    Two-short flags extend 4 to the left and 4 to the right no matter in which state they are.
    Observe that the horizontal distance between striker axes of two neighboring strikers is~15.
    Therefore, for each pair of long flags of two neighboring strikers that are at the same height
    (we say in the same \emph{clause strip}),
    the left flag cannot be in the state right while the right flag is in the state left;
    otherwise their long sides would overlap.
    The distance from the frame to the striker axis of the leftmost and the rightmost striker is~7.
    Hence, all long flags of the leftmost striker can only be in the state right,
    while all long flags of the rightmost striker can only be in the state left.
    These properties together imply that there cannot be $n$ long strikers in the same clause strip.

    \paragraph{Correctness.}
    It remains to argue that there is a crossing-free drawing of the \PPG~$G$ with unit edge length
    if and only if the input \textsc{NotAllEqual3SAT} instance~$F$ is a yes-instance.
    
    First assume that $F$ is a yes-instance.
    Take a corresponding assignment $\Phi$ of truth values to variables.
    Draw $G$ with only edge length~1 as illustrated in \cref{fig:nph}, that is,
    for each variable that is assigned true in~$\Phi$,
    draw its striker in the base orientation (shaft part in the state down) and,
    for each variable that is assigned false in~$\Phi$,
    draw its striker in the flipped orientation (shaft part in the state up).
    Consequently, we have in every clause strip at least one flag with two short sides;
    let $f$ be such a short flag.
    Draw all flags on the left of $f$ in the state right,
    and draw all flags on the right of $f$ in the state left.
    Overall, this yields a crossing-free drawing with edge length~1 only.
    
    Now assume that there is a crossing-free drawing of the \PPG~$G$ with edge length~1 only.
    There is only one way to draw a 4-cycle with unit length, that is, as an integer grid square.
    This results in the rigid structure of the frame, shaft, flags, and strikers described before.
    Without loss of generality, assume that the shaft extends along a horizontal axis
    as illustrated in \cref{fig:nph}
    (if it extends along a vertical axis, mirror the drawing along a line of 45 degrees).
    The only remaining flexibility are the states of the shaft parts and the flags,
    as well as the drawing of the long subpaths of shared edges connecting the top-/bottommost flag
    with the shaft part of the striker.
    
    We next argue that we can read off a truth assignment for $F$
    that implies that $F$ can only be a yes-instance.
    We get such an assignment by setting all variables to true whose striker
    is drawn in the base orientation (shaft part in the state down) and
    by setting all other variables to false.
    Due to the width of the long sides of the flags and the width of the gaps between the striker axes and the frame,
    every long side of a flag in a clause strip must be in its own gap.
    Since there are only $n-1$ sufficiently wide gaps per clause strip,
    there is a flag with two short sides in each clause strip, which corresponds to a literal satisfying this clause~--
    once in the clause strip above the shaft, once in the clause strip below the shaft.
    These two clause strips together
    assure that not all literals in a clause have the same truth value.
\end{proof}

\end{document}